\documentclass[aps,pra,reprint,superscriptaddress]{revtex4-1}

\usepackage[utf8]{inputenc}

\usepackage{amsthm,amsmath,latexsym,amssymb,verbatim,enumerate,graphicx}
\usepackage[svgnames,psnames]{xcolor}
\usepackage[colorlinks,
citecolor=grayblue,
linkcolor=grayblue,
urlcolor=grayblue,
linktocpage,
unicode]{hyperref}
\usepackage{color}
\usepackage{tikz}
\usepackage{subeqnarray}
\usepackage{times,txfonts}
\usepackage{nicefrac}
\usepackage{ragged2e}
\usepackage{enumitem}
\usepackage{blindtext}

\DeclareMathAlphabet{\mathcal}{OMS}{cmsy}{m}{n} 

\definecolor{darkblue}{rgb}{0.0, 0.0, 0.35}
\definecolor{grayblue}{rgb}{0.10, 0.45, 0.6}
\definecolor{darkgreen}{rgb}{0.0, 0.3, 0.0}
\definecolor{darkpurple}{rgb}{0.3, 0.0, 0.25}
\definecolor{darkorange}{rgb}{1.0, 0.549, 0.0}
\definecolor{darkred}{rgb}{0.5, 0.0, 0.0}
\definecolor{verylightgray}{rgb}{0.9, 0.9, 0.9}
\definecolor{gray75}{gray}{0.3}

\hypersetup{colorlinks=true, citecolor=darkblue, linkcolor=darkblue, urlcolor=darkblue}

\definecolor{orcidlogocol}{HTML}{A6CE39}
\tikzset{
  orcidlogo/.pic={
    \fill[orcidlogocol] svg{M256,128c0,70.7-57.3,128-128,128C57.3,256,0,198.7,0,128C0,57.3,57.3,0,128,0C198.7,0,256,57.3,256,128z};
    \fill[white] svg{M86.3,186.2H70.9V79.1h15.4v48.4V186.2z}
                 svg{M108.9,79.1h41.6c39.6,0,57,28.3,57,53.6c0,27.5-21.5,53.6-56.8,53.6h-41.8V79.1z M124.3,172.4h24.5c34.9,0,42.9-26.5,42.9-39.7c0-21.5-13.7-39.7-43.7-39.7h-23.7V172.4z}
                 svg{M88.7,56.8c0,5.5-4.5,10.1-10.1,10.1c-5.6,0-10.1-4.6-10.1-10.1c0-5.6,4.5-10.1,10.1-10.1C84.2,46.7,88.7,51.3,88.7,56.8z};
  }
}

\newcommand\orcidicon[1]{\href{https://orcid.org/#1}{\mbox{\scalerel*{
\begin{tikzpicture}[yscale=-1,transform shape]
\pic{orcidlogo};
\end{tikzpicture}
}{|}}}}

\newcommand{\dd}{\mathrm{d}}		                     									

\newcommand{\hsp}{\hspace{0.05cm}}															

\newcommand{\ii}{\mathrm{i}}
\newcommand{\ee}{\mathrm{e}}
\newcommand{\TT}{\mathrm{T}}
\newcommand{\bra}[1]{\langle#1|}
\newcommand{\ket}[1]{|#1\rangle}

\newcommand{\tr}{\mathrm{Tr}\hspace{0.05cm}}

\newcommand{\Ai}{\mathrm{A}_{\tt i}}
\newcommand{\Ao}{\mathrm{A}_{\tt o}}

\newcommand{\Bi}{\mathrm{B}_{\tt i}}
\newcommand{\Bo}{\mathrm{B}_{\tt o}}

\newcommand{\pos}{\succcurlyeq}

\DeclareMathOperator*{\argmax}{arg\,max}

\newtheorem{thm}{Theorem}
\newtheorem{lem}[thm]{\bf Lemma}
\newtheorem{prop}[thm]{\bf Proposition}
\newtheorem{propsm}[thm]{\bf Proposition}
\newtheorem{obs}{\bf Observation}

\newtheorem{defn}{\bf Definition}

\theoremstyle{remark}																		

\newtheorem*{rem}{Remark}

\begin{document}

\title[Quantifying information flow in quantum processes]{Quantifying information flow in quantum processes}

\author{Leonardo S. V. Santos}
\email{leonardo.svsantos@student.uni-siegen.de}
\affiliation{Naturwissenschaftlich-Technische Fakultät, Universität Siegen, Walter-Flex-Straße 3, 57068 Siegen, Germany}
\author{Zhen-Peng Xu}
\email{zhen-peng.xu@ahu.edu.cn}
\affiliation{School of Physics and Optoelectronics Engineering, Anhui University, 230601 Hefei, China}
\affiliation{Naturwissenschaftlich-Technische Fakultät, Universität Siegen, Walter-Flex-Straße 3, 57068 Siegen, Germany}
\author{Jyrki Piilo}
\email{jyrki.piilo@utu.fi}
\affiliation{Department of Physics and Astronomy, University of Turku, FI-20014, Turun Yliopisto, Finland}
\author{Otfried Gühne}
\email{otfried.guehne@uni-siegen.de}
\affiliation{Naturwissenschaftlich-Technische Fakultät, Universität Siegen, Walter-Flex-Straße 3, 57068 Siegen, Germany}

\begin{abstract}
We present a framework for quantifying information flow within general quantum processes. For this purpose, we introduce the signaling power of quantum channels and discuss its relevant operational properties. This function supports extensions to higher-order maps, enabling the evaluation of information flow in general quantum causal networks and also processes with indefinite causal order. Furthermore, our results offer a rigorous approach to information dynamics in open systems that applies also in the presence of initial system-environment correlations, and allows for the distinction between classical and quantum information backflow.
\end{abstract} 

\date{\today}

\maketitle

\section*{Introduction}

A fundamental challenge in information theory lies in characterising the exchange of information between agents.~Prominent examples include communication tasks, where a sender encodes a message and transmits it through a channel to a receiver, as well as the storage of information in a memory for future retrieval. Ultimately, information processing occurs in physical systems governed by the laws of classical or quantum physics, always under the influence of some type of noise~\cite{Landauer91}.

Recent years have seen remarkable advances in quantum information science through the modelling of physical processes via high-order operations~\cite{Taranto2025}. This theoretical framework describes in full generality the quantum resources that are feasible within the constraints of causality. In essence, it allows efficient optimisation in quantum computing~\cite{CDP08b} and metrology~\cite{Chiribella12,Bavaresco2024}, as well as the operational modelling of quantum stochastic processes~\cite{PRFPM18}. From a foundational standpoint, such maps have proven to be powerful tools for exploring the quantum world in situations where the causality of events defies our common sense~\cite{OCB12}, for example, when the order~\cite{CDPV13} or the temporal direction~\cite{CL22} of transformations is coherently superposed.

This article introduces an operational framework for characterizing information flow between agents sharing arbitrary quantum resources. To this end, we employ the formalism of higher-order quantum operations. Central to our approach is the definition of the signalling power of a quantum channel, a function that possesses a broad set of desirable operational properties, admits efficient computation via semidefinite programming, and has a clear operational interpretation as the average success probability in a superdense-coding-like protocol (see Fig.~\ref{fig:1}).

Building on this, we extend the concept of information backflow~\cite{BLP09}, a central concept in the theory of quantum non-Markovianity~\cite{BLPV16}, in several directions. Notably, our framework remains well-defined even in the presence of initial system-environment correlations, a regime in which conventional approaches based on dynamical maps typically fail. Furthermore, we introduce the notion of quantum information backflow as a witness of genuine quantum memory: specifically, the observation of quantum information backflow indicates memory effects that cannot be simulated classically.


\begin{figure}
    \centering
    \includegraphics[width=1\linewidth]{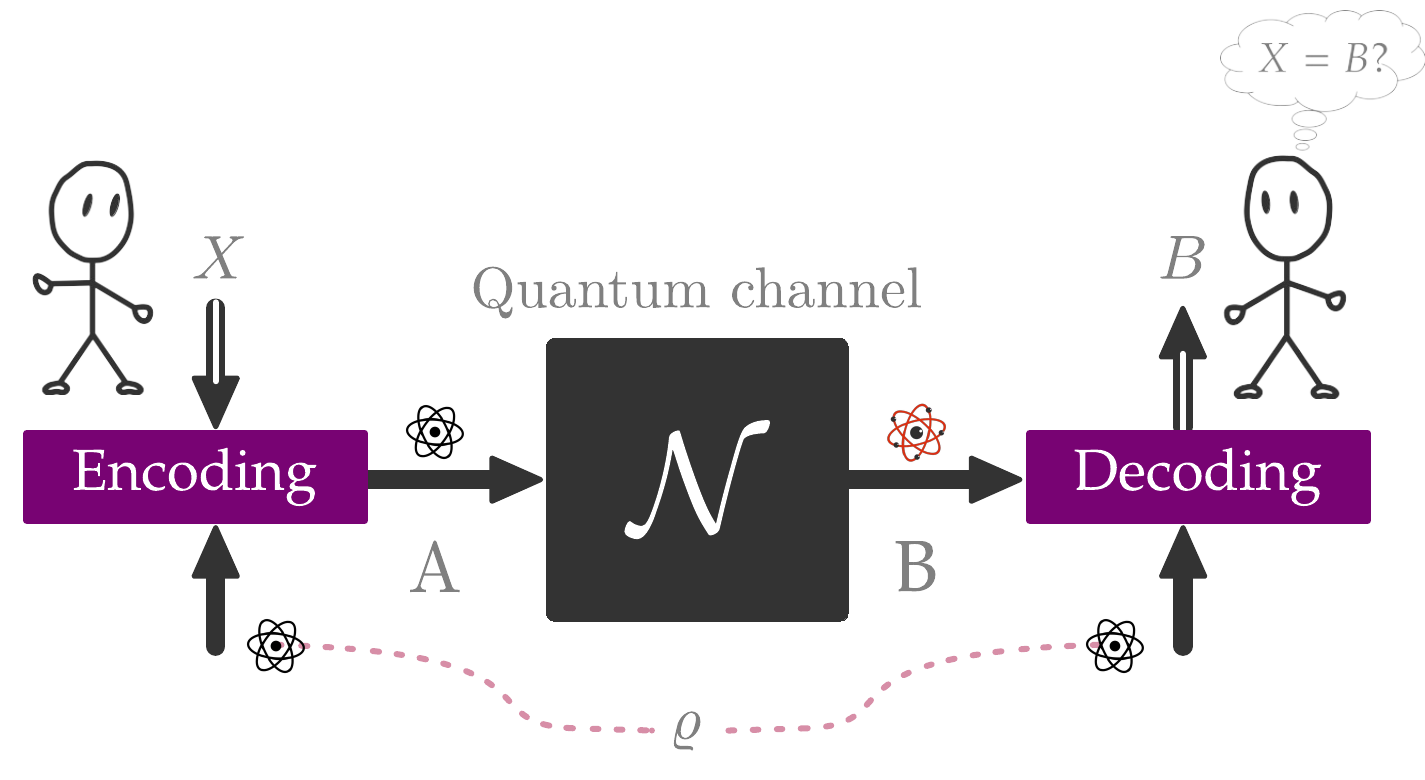}
    \caption{Communication scenario where two agents aim to transmit classical information (variables $X$ and $B$) using quantum systems: the agents have available a quantum channel $\mathcal{N}$ to exchange quantum systems and potentially share classical randomness and/or entanglement in the form of a shared bipartite state $\varrho$. In this scenario, the quantum channel $\mathcal{N}$ defines what we call a ``quantum process''. Although illustrative, in this article we consider more general quantum processes too; see Figure \ref{fig:2}.}
    \label{fig:1}
\end{figure}

\section*{Results}

{\bf Setup.---} We start by fixing our notation. Hereafter, $\mathcal{H}_{\rm X}$ denotes the finite-dimensional Hilbert space for the quantum system labelled by $\rm X$.
The symbols $\mathrm{id}_{\rm X}$, $\tr_{\rm X}$, and $\Gamma_{\rm X}$ refer to the identity matrix, partial trace and partial transposition over $\rm X$, respectively. Finally, we often invoke the Choi-Jamiółkowski (CJ) isomorphism \cite{Choi75,J72} to represent quantum transformations both as linear maps and as non-normalised states. We explain the theory behind it in detail in \hyperlink{AppendixA}{Appendix A}.

We consider processes in which the values of two classical variables $A$ and $B$ are generated by measurements performed on quantum systems (Fig.~\ref{fig:2}). It is illustrative to think of such variables as representing the outcomes of measurements performed by two agents, say Alice (A) and Bob (B), residing in separated laboratories (labs). To generate their outcomes, they may exchange quantum systems, and share non-signalling resources (entanglement or classical randomness) encoded in a bipartite quantum state. Moreover, both may have access to local randomness sources that are independent of any resource they share. Informally, that is to say that the agents have a ``free-will'' to decide which measurement to perform.

\begin{figure}
    \centering
    \includegraphics[width=0.65\linewidth]{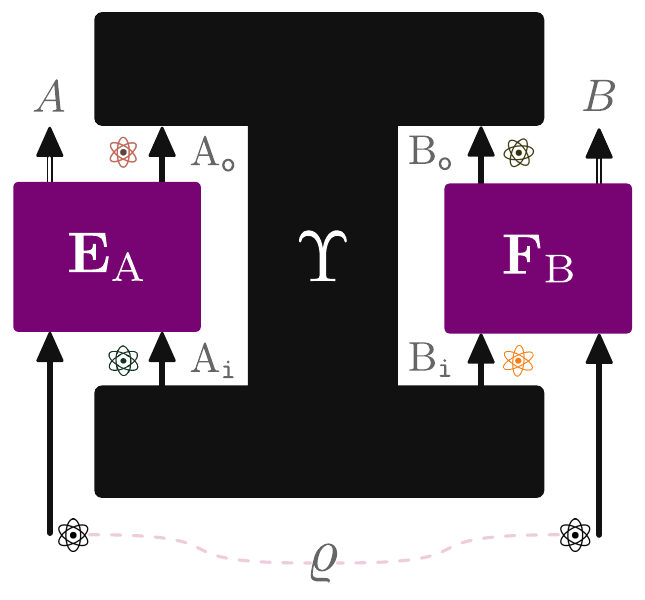}
    \caption{Depiction of the general setup we consider [cf. Eq.~\eqref{eq:g_born}]: Two agents manipulate quantum systems locally through quantum instruments, with $A$ and $B$ representing the measurement outcomes. $\Upsilon$ defines the process matrix, that is, the most general communication resource they can use. As in Fig.~\ref{fig:1}, they may share a quantum state $\varrho$, which encodes entanglement and/or classical randomness.}
    \label{fig:2}
\end{figure}

In a broad perspective, one may only assume that the quantum probabilistic framework is locally valid inside each lab~\cite{OCB12}.~Each experimenter can perform a measurement on the systems that he/she has access to, which is formally described by a suitable quantum instrument; that is, a general measurement where the post-measurement state is also taken into account. Mathematically, each outcome is associated to a completely-positive (CP) map, such that the sum of these maps is trace-preserving (TP). The probability of each outcome is computed by taking the trace of the state after applying the corresponding CP map.

In this general scenario, the measurement statistics are obtained via~\cite{OCB12}
\begin{equation}\label{eq:g_born}
\mathrm{Pr}(AB\mid \mathbf{E}_{\rm A}\mathbf{F}_{\rm B})=\tr[(E_A^\TT \otimes F_B^\TT)(\Upsilon\otimes \varrho)].
\end{equation}
Here, $\mathbf{E}_{\rm A}=\{E_{A=a}\}$ and $\mathbf{F}_{\rm B}=\{F_{B=b}\}$ are instruments (in the CJ picture) of Alice and Bob, $\varrho$ is a bipartite state, and the operator $\Upsilon$ is the so-called \emph{process matrix}, describing the communication resources connecting the labs (see Fig.~\ref{fig:2}). If $\Ai$($\Bi$) and $\Ao$($\Bo$) label the in- and output systems of Alice(Bob), then $\Upsilon$ is an operator on the tensor product space $\mathcal{H}_{\Ai}\otimes \mathcal{H}_{\Ao}\otimes \mathcal{H}_{\Bi}\otimes\mathcal{H}_{\Bo}$. With no further assumptions, $\Upsilon$ can be any operator for which Eq.~\eqref{eq:g_born} gives rise to non-negative and normalised probabilities. All possible quantum resources the agents may use for communication can be described by such an operator, including channels, channels with memory, and also processes with indefinite causal order. Each measurement operator, $E_A$ or $F_B$, acts on one part of the bipartite state $\varrho$ and in the corresponding in- and output spaces of each lab.

Given this scenario, it is natural to ask for signalling and information flow between the agents. Intuitively, the more sensitive Bob's local statistics are to Alice's encodings, the better she can 
send information to him. More formally, Alice can\emph{not} signal to Bob when his local 
measurement statistics are indifferent to Alice's setting; that is, when the marginal 
distribution
\begin{equation}\label{eq:non_signalling}
\mathrm{Pr}(B\mid \mathbf{E}_{\rm A}\mathbf{F}_{\rm B})=\sum_a \mathrm{Pr}(A=a,B\mid \mathbf{E}_{\rm A}\mathbf{F}_{\rm B})
\end{equation}
is independent on the instrument $\mathbf{E}_{\rm A}$. Equivalently, suppose they have agreed, \emph{a priori}, with a finite alphabet of alternatives $\mathcal{X}$. Alice encodes a letter $x\in\mathcal{X}$ into some system by applying a channel $\mathcal{E}_x$, and then sends it out of her lab. If Alice cannot signal to Bob, then his local measurement statistics are indifferent to Alice's encoding, i.e., the equality $\mathrm{Pr}(B\mid \mathcal{E}_x,\mathbf{F}_{\rm B})=\mathrm{Pr}(B\mid \mathcal{E}_{x^\prime},\mathbf{F}_{\rm B})$ holds for \emph{all} pair of channels, $\mathcal{E}_{x}$ and $\mathcal{E}_{x^\prime}$.

{\bf The signalling power of a channel.---} The simplest communication resource Alice can use to communicate quantum systems to Bob is a quantum channel. Mathematically, a channel is represented by a CPTP linear map $\mathcal{N}$ or, in the CJ picture, by a positive semidefinite operator $N$ satisfying $\tr_{\rm B} \hspace{0.05cm}N=\mathrm{id}_{\rm A}$.

If the labs are connected through a channel $N$, Eq.~\eqref{eq:g_born} reduces to
\begin{equation}\label{eq:g_born2}
\mathrm{Pr}(AB\mid \mathbf{E}_{\rm A}\mathbf{F}_{\rm B})=\tr[(E_A^\TT \otimes F_B^\TT)(N\otimes \varrho)].
\end{equation}
In this scenario, we can consider that Bob discards (traces-out) his system after measurement. Consequently, $\mathbf{F}_{\rm B}$ can be considered as a positive-operator-value-measure (POVM). Notice that Eq.~\eqref{eq:non_signalling} is satisfied for all POVMs if and only if $N=\mathrm{id}_{\rm A}\otimes \varpi$ for some state $\varpi$; equivalently, the channel traces-out the input state and prepare a system in a state $\varpi$ fixed \emph{a priori}, $\mathcal{N}(\cdot)=\tr(\cdot)\varpi$.

To assess the ``quality'' of a channel for signalling, we introduce its \emph{signalling power} as
\begin{equation}\label{eq:SP}
\mathrm{S}(N)=\log_2\max_{W} \tr(N W),
\end{equation}
with maximisation to be taken over operators $W$ representing channels in the opposite causal direction of the channel $N$, i.e., $W\pos 0$ and $\tr_{\rm A}\hspace{0.05cm}W=\mathrm{id}_{\rm B}$. The motivation of this definition will become clear after, where we demonstrate its connection with the superdense-coding protocol; moreover, it satisfies a plethora of properties that are required for a quantifier of the signalling power. At the moment, it is worth noting that $\rm S$ is evaluated as a semidefinite program and, as such, can be efficiently computed with standard methods employed in quantum information science~\cite{SC23}. Furthermore, $\rm S$ is \emph{equivalent} to (up to normalization and sign) the conditional min-entropy of the CJ state of the channel (see, e.g., Ref.~\cite{Konig2009}).

{\bf Properties of the signalling power.---} We will now establish the main properties and operational interpretation of the signalling power (see \hyperlink{AppendixB}{Appendix B} for formal proofs).


\begin{obs}\label{obs:1}
The signalling power [Eq.~\eqref{eq:SP}] satisfies the following properties:
\begin{itemize}
    \item[\rm (a)] $\mathrm{S}(N)$ is a convex and non-negative function of $N$;
    \item[\rm (b)] $\rm S$ is continuous with respect to the diamond norm;
    \item[\rm (c)] $\mathrm{S}(N)=0$ if and only if $N=\mathrm{id}_{\rm A}\otimes \varpi$ for some density 
    matrix $\varpi$;
    \item[\rm (d)] $\rm S$ is additive, i.e., $\mathrm{S}(N\otimes M)=\mathrm{S}(N)+\mathrm{S}(M)$.
\end{itemize}
\end{obs}




Concerning the operational interpretation, we make use of the superdense coding protocol \cite{BW92}. 
For simplicity, consider a qudit channel (i.e., input and output having the same dimensionality $d$).
Alice and Bob agreed \emph{a priori} with $d^2$ alternatives, say $\mathcal{X}=\{(k\ell)\}_{0\leq k,\ell\leq d-1}$, and their communication is assisted by a maximally entangled state, $\ket{\Phi_d^+}\propto\sum_{i}\ket{i}\otimes \ket{i}$, where $\{\ket{i}\}$ is a fixed orthonormal basis. Alice encodes a letter $x=(k\ell)$ by applying the corresponding Weyl unitary $V_{k\ell}=X_d^k Z_d^\ell$, where $X_d$ and $Z_d$ are defined via $X_d\ket{j}=\ket{j+1\,\,\mathrm{mod}\, d}$ and $Z_d\ket{j}=\ee^{\ii 2\pi j/d}\ket{j}$. If the communication channel is unitary, $\mathcal{N}(\cdot)=U\cdot U^\dagger$, Bob then receives one of $d^2$ states, $UV_{k\ell}\ket{\Phi_d^+}$, which are orthogonal by construction and hence can be distinguished in a single-shot measurement. Consequently, two dits of information were perfectly transmitted. Moreover, $2^{\mathrm{S}(N)}=d^2$ for unitary channels. For a noisy channel, instead, one has the following result:

\begin{obs}\label{obs:2}
For a channel $N$, used for communication in the superdense-coding protocol assisted by a maximally entangled state, one has
\begin{equation}\label{eq:expS_coincidence}
2^{\mathrm{S}(N)}=\max\sum_{x\in \mathcal{X}}\mathrm{Pr}(B=X|X=x),
\end{equation}
where $|\mathcal{X}|=(\dim\mathcal{H}_{\rm A})^2$, $\mathrm{Pr}(B=X|X=x)$ is the probability of correctly decoding the input $x$, and the maximum is taken over all possible encoding and decoding strategies.
\end{obs}

The effectiveness of the superdense coding protocol relies on entanglement \cite{HP12,MNGRC21}. In the ideal scenario, it is preserved as Alice's encoding and the communication channel act like a local unitary transformation. For the case of an entanglement-breaking channel, that is, a quantum channel that when applied to a part of any bipartite state makes it separable, one has the following:

\begin{obs}\label{prop:entanglement_breaking}
For entanglement-breaking channels, the signalling power is upper bounded 
by $\log_2(\dim\mathcal{H}_{\rm A})$.
\end{obs}




\begin{figure}
    \centering
    \includegraphics[width=0.65\linewidth]{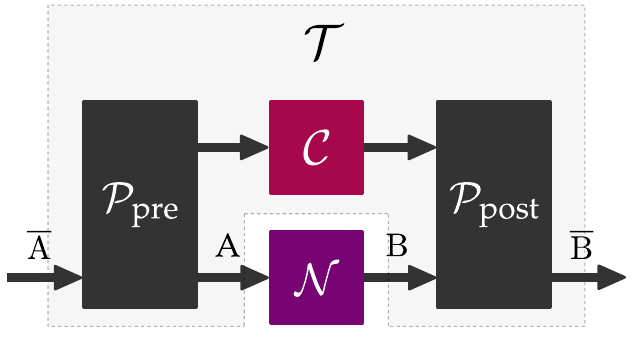}
    \caption{Quantum superchannel $\mathcal{T}$ transforming a channel $\mathcal{N}_{\rm A\to B}$ into $\hat{\mathcal{N}}_{\rm \bar{A}\to \bar{B}}=\mathcal{T}(\mathcal{N})$. It is defined in terms of a pre- and post-processing channels that linked through a memory described by a channel $\mathcal{C}$.}
    \label{fig:supermap}
\end{figure}

{\bf A data-processing inequality.---} Next, we discuss how different channel transformations affect the signalling power, in particular we aim to state a general data-processing inequality (DPI). First, recall that a channel can be transformed in many different ways, with appropriated pre- and post- processing operations that are potentially linked through an auxiliary memory system (see Fig.~\ref{fig:supermap}). Mathematically, a linear map $\mathcal{T}$ represents an admissible channel transformation (or superchannel) if it maps channels into channels, even when applied locally on a bipartite channel~\cite{CDP08a}. As illustrated in Fig.~\ref{fig:supermap}, any superchannel can be decomposed as
\begin{equation}\label{eq:admissible_maps}
\mathcal{T}[\mathcal{N}]=\mathcal{P}_{\rm post}\circ (\mathcal{N}\otimes \mathcal{C})\circ \mathcal{P}_{\rm pre},
\end{equation}
In the CJ picture, we write
\begin{equation}
    \mathcal{T}[\mathcal{N}]=T\star N,
\end{equation}
where $T$ denotes the CJ representation of $\mathcal{T}$ and $\star$ the link product \cite{CDP09}. The link product is formally defined in \hyperlink{AppendixA}{Appendix A}. Without prejudice to the understanding of the rest of this article, this equation can be read simply as a notation.

The connection between pre- and post-processing through the channel $\mathcal{C}$ defines the \emph{memory} of the high-order transformation $\mathcal{T}$. In a memoryless transformation, such maps are independent of each other; equivalently, $\mathcal{C}$ is a trace-and-prepare channel. Classical superchannels are those presenting only classical memory (no quantum memory). Formally, this is the same as saying that $\mathcal{C}$ is a measure-and-prepare (entanglement-breaking) channel  \cite{RBL18}. If the superchannel cannot be simulated in this manner, we say that the transformation $\mathcal{T}$ has quantum memory.

To state the DPI for $\rm S$, we first recall its version for quantum states: the entropy of a quantum state cannot decrease under the action of a bistochastic channel. A channel is bistochastic if it has the maximally mixed state as a fixed point. In particular, if input and output have the same dimensionality, they satisfy $\tr_{\rm B}\hspace{0.05cm}N=\mathrm{id}_{\rm A}$ and $\tr_{\rm A}\hspace{0.05cm}N=\mathrm{id}_{\rm B}$; that is, $N$ represents at the same time a channel from $\rm A$ to $\rm B$ and a channel from $\rm B$ to $\rm A$. A superchannel, instead, maps channels from $\rm A$ to $\rm B$ into channels from, say, $\rm \bar{A}$ to $\rm \bar{B}$. We then define the CJ operator $T$ of a superchannel as bistochastic if it also represents an admissible transformation of channels from $\rm B$ to $\rm A$ into channels from $\rm \bar{B}$ to $\rm \bar{A}$. In essence, $T$ is bistochastic if it represents two causally inverse channel transformations at the same time. 

\begin{obs}\label{prop:DPI}
For a given superchannel with CJ operator $T$, the following statements are equivalent:
\begin{itemize}
    \item[(a)] $\mathrm{S}(T\star N)\leq \mathrm{S}(N)$ for all input channel $N$.
    \item[(b)] The superchannel $T$ maps non-signalling channels into non-signalling channels.
    \item[(c)] $T$ is bistochastic.
\end{itemize}
\end{obs}

Bistochastic superchannels are in one-to-one correspondence with non-signalling bipartite channels (i.e., a channel with two distinct inputs and two outputs). A bipartite channel describes a deterministic input-output operation—sometimes referred to as a ``box''—that acts on the systems accessible to two agents, Alice and Bob. In our setting, $\mathcal{T}$ corresponds to a bipartite channel with inputs $\rm \bar{A}\, B$, and outputs $\rm A\, \bar{B}$, where Alice has access to the pair $\rm (A,\, \bar{A})$ and Bob to $\rm (B,\, \bar{B})$.

An interesting consequence of this correspondence is that superchannels with no communication between pre- and post-processing operations form a \emph{strict subset} of bistochastic superchannels. This is because there exist non-signalling boxes (such as the well-known PR-box~\cite{PR1994}) that, while non-signalling, cannot be realized using shared quantum states alone. In particular, the PR-box violates information causality~\cite{P2009}: when interpreted as a superchannel, it can increase the classical capacity of certain input channels. This suggests that there might exist hidden fine-tuned information flow that is not captured by violation of DPI alone. Seen through our lens, recent breakthroughs on information causality~\cite{Miklin2021,Jain2024} show that one needs an even higher order operation (a superchannel transformation or ``super-superchannel'') to ``activate'' this hidden memory (in the jargon typically employed in such works, the super-superchannel is what defines the ``protocol'').

{\bf Connections with anti-distinguishability.---} It is natural to wonder what would happen if instead of maximisation we had minimisation in the definition of signalling power. By doing so, we obtain another signalling quantifier, which also has interesting properties.

One way to consistently replace minimisation with maximisation is to define
\begin{equation}\label{eq:P}
\mathrm{P}(N)=1-\min_{W}\tr(NW).
\end{equation}
This function retains several properties of $\mathrm{S}$: it is non-negative, convex, faithful, continuous, and satisfies the DPI. Furthermore, it is related to $\mathrm{S}$ via
\begin{equation}
\mathrm{P}(N)=[\dim(\mathcal{H}_{\rm A}\otimes \mathcal{H}_{\rm B})-1]\big(2^{\mathrm{S}(\hat{N})}-1\big),
\end{equation}
where
\begin{equation}
\hat{N}=\frac{(\dim\mathcal{H}_{\rm A})\mathrm{id}_{\rm AB}-N}{\dim(\mathcal{H}_{\rm A}\otimes \mathcal{H}_{\rm B})-1}.
\end{equation}
However, in contrast to the signalling power, $\mathrm{P}$ does not present any separation between classical and quantum, in the sense that there is no value $\gamma_{\rm Cl}$ such that $\mathrm{P}(N)\geq \gamma_{\rm Cl}$ implies non-classicality (see Observation~\ref{prop:entanglement_breaking}). In particular, there are channels that are entanglement-breaking channels maximising $\mathrm{P}(N)$. For example, the qubit channel represented by
\begin{equation}
N_{\rm EB}=\ket{0}\bra{0}\otimes \ket{0}\bra{0}+\ket{1}\bra{1}\otimes \ket{1}\bra{1},
\end{equation}
is such that $\mathrm{P}(N_{\rm EB})=1$, which is the maximum value $\rm P$ can assume. Moreover, this function is \emph{superadditive}. In particular, we numerically found the channel represented by
\begin{equation*}
\textrm{\normalsize $\tilde{N}=$}{\tt\left(
\begin{array}{cccc}\tt
 0.70836 & \tt 0.23062 &\tt -0.24562 &\tt -0.07939 \\
 \tt 0.23062 &\tt 0.29164 &\tt 0.26927 &\tt 0.24562 \\
 \tt -0.24562 &\tt 0.26927 &\tt 0.64901 &\tt 0.46428 \\
\tt -0.07939 &\tt 0.24562 &\tt 0.46428 &\tt 0.35100 \\
\end{array}
\right)},
\end{equation*}
which satisfies $\mathrm{P}(\tilde{N}) < 1$, yet $\mathrm{P}(\tilde{N}^{\otimes 2})\overset{(!)}{=}1$ (up to numerical precision). Even though we focus more on the signalling power throughout this paper, this function will be relevant when we discuss the conditions for information backflow in terms of the time-dependent decay rates in non-Markovian dynamics and the bipartite scenario with indefinite causal order as we will see below.

It is easy to verify that $\mathrm{P}(N) \leq 1$, with the maximum attained by channels that permit perfect alternative exclusion in a superdense-coding-like protocol. The scenario is analogous to standard superdense coding (see Fig.~\ref{fig:1}), except that the receiver now aims to \emph{exclude} one possible encoding option with certainty rather than correctly identifying the encoded message. In this context, $\mathrm{P}(N) = 1$ if and only if there exists a set of states, $$\{\varrho_x=(\mathcal{N}\circ\mathcal{M}_x)\otimes \mathcal{I}_{\rm \bar{B}}[\Phi^+]:x\in\mathcal{X}\},$$ that is perfectly \emph{anti-distinguishable}; that is, without any prior information one can perform a measurement that correctly excludes with probability one at least one of them as the correct one. More generally, $\mathrm{P}(N)$ equals
\begin{equation}\label{eq:Pantidist}
\mathrm{P}(N) = 1 - |\mathcal{X}| + \max \sum_{x \in \mathcal{X}} \Pr(B \neq X \mid X = x),
\end{equation}
which can be easily derived by adapting the proof of Observation~\ref{obs:2} described in the \hyperlink{AppendixB}{Appendix B}.

{\bf Application to open quantum systems.---} Within the context of open quantum system (OQS) dynamics, deviations from DPI are signatures of non-Markovianity. The fundamental idea is that, in a Markovian evolution, all information that is encoded on the OQS state at a given time is monotonically lost to the environment as they interact  \cite{BLP09}. If an OQS is initialised (here $t=0$) in a state that is uncorrelated from its environment, its evolution is described by a quantum dynamical map, namely, a family of CPTP maps $\mathcal{N}_t$ (equivalently, their CJ operators $N_t$). The absence of information backflow (with respect to $\rm S$) at a time $t\geq 0$ reads
\begin{equation}\label{eq:no_info_backflow}
-\frac{\dd}{\dd t}\mathrm{S}(N_t)\geq 0.
\end{equation}
As a consequence of Observation~\ref{prop:DPI}, this is always satisfied for CP-divisible dynamical maps, namely, those for which it is possible to define quotient CPTP maps $\mathcal{N}_{t|s}$ such that $\mathcal{N}_t=\mathcal{N}_{t|s}\circ \mathcal{N}_s$ [e.g., the semigroup dynamics generated by a Lindbladian $\mathcal{L}$, $\mathcal{N}_t=\ee^{t\mathcal{L}}$, where $\mathcal{N}_{t|s}=\ee^{(t-s)\mathcal{L}}$]. This remains valid if we exchange $\rm S$ by $\rm P$.

To illustrate how non-Markovianity can be detected through the breakdown of the inequality \eqref{eq:no_info_backflow}, we consider the evolution of a two-level system under the action of a phase-covariant noise, which includes damping and dephasing. The master equation for such a system takes the form~\cite{SKHD16}
\begin{align}
\dot{\varrho}_t&=-\ii \omega(t)[\sigma_z,\varrho_t]\nonumber \\ &\quad+\sum_{i} \frac{\gamma_i(t)}{2}\left(\sigma_i\varrho_t \sigma_i^\dagger-\frac{1}{2}\{\sigma_i^\dagger \sigma_i,\varrho_t\}\right),
\end{align}
$i\in\{+,-,z\}$, $\sigma_z=\ket{1}\bra{1}-\ket{0}\bra{0}$, and $\sigma_-^\dagger=\sigma_+=\ket{1}\bra{0}$. CP-divisibility is broken whenever any of the decay rates is negative~\cite{RHP14}. 

We derived conditions for no information backflow as quantified by the functions we introduced (see \hyperlink{AppendixC}{Appendix C} for details), obtaining
\begin{align}\label{eq:no_info_backflow_decay}
\gamma(t)&+4\gamma_z(t)\nonumber \\ &\pm \gamma(t)\exp\left(\frac{1}{4}\int_0^t[4\gamma_z(s)-\gamma(s)]\dd s\right)\geq 0,
\end{align}
$\gamma(t)=\gamma_+(t)+\gamma_-(t)$, the condition with the minus sign applying to cases where the argument of the exponential is positive, being replaced by the positivity of $\gamma(t)$ otherwise. Therefore, the signalling power differs from all other commonly used quantifiers (see Ref.~\cite{TLSM18}) by having a ``nonlocal-in-time'' term. Furthermore, if $\gamma(t)=\gamma$, a positive constant, then the function $\gamma_z(t)$ for which the left-hand-side of Eq.~\eqref{eq:no_info_backflow_decay} is null for all times coincides exactly with so-called ``eternally non-Markovian dynamics", where CP-divisibility is broken at all points of time  \cite{HCLA14}; that is, the solution for corresponding integral equation for $\gamma_z(t)$, for the equality in Eq.~\eqref{eq:no_info_backflow_decay} with the minus sign reads
\begin{equation}
\gamma_z(t)=-\frac{\gamma}{4}\tanh\left(\frac{\gamma t}{4}\right).
\end{equation}
Finally, we compare the detection of CP-divisibility breakdown of signalling power to other approaches. For this, we look at the following simple model~\cite{TM21}:
\begin{subeqnarray}
\gamma_+(t)&=&\ee^{-t/2}, \\
\gamma_-(t)&=&\ee^{-t/4}, \\
\gamma_z(t)&=&\frac{\kappa\ee^{-3t/8}}{2} \cos(2t),
\end{subeqnarray}
where $\kappa\geq 0$. CP-divisibility is broken for all $\kappa>0$ [since $\gamma_z(t)$ assumes negative values], while P-divisibility is broken only if $\sqrt{\gamma_+(t)\gamma_-(t)}+2\gamma_z(t)<0$~\cite{FGL20}, which occurs when $\kappa>1$. Information backflow, as quantified with the trace-distance, is detectable only if $\gamma_+(t)+\gamma_-(t)+4\gamma_z(t)<0$. For the signalling power we numerically verify that it significantly enhances the detection compared to the trace distance; namely, we found
\begin{align}
\kappa&\overset{\text{(TD)}}>\underbrace{\cosh(\pi/16)}_{\approx1.01934}\overset{\text{(SP)}}{>}0.44044,
\end{align}
where TD and SP stand for ``trace distance'' and ``signalling power'', respectively.

Now we go beyond the standard paradigm of dynamical maps to describe open system information dynamics. For this, we define quantum dynamical supermaps as
\begin{equation}\label{eq:super_dynamical_map}
\mathcal{T}_{t,\hspace{0.025cm}s}[\mathcal{N}](\cdot)=\tr_{\rm Env}\hsp\mathcal{U}_{t|s}(\mathcal{N}\otimes \mathcal{I}_{\rm Env}) \mathcal{U}_{s|0}(\hspace{0.05cm}\cdot\hspace{0.05cm} \otimes \varrho_{\rm Env}),
\end{equation}
where $\mathcal{U}_{t|s}$ is the unitary map acting on the total system (i.e., System+Env) and $\varrho_{\rm Env}$ is the initial state for the environment. For two given times $s$ and $t$, $0\leq s\leq t$, the map $\mathcal{T}_{t,s}$ 
transforms $N$ into $T_{t,s}\star N$ through Eq.~(\ref{eq:super_dynamical_map}). If $T_{t,s}$ is \emph{not} bistochastic, then for some channel, $\mathrm{S}(T\star N)>\mathrm{S}(N)$. This is arguably a signature of information backflow: since the total system is isolated, the only way the process described by $\mathcal{T}_{t,s}$ can improve signalling of a noise channel is through transmitting and retrieving information through the environment.

\begin{figure}
    \centering
    \includegraphics[width=0.8\linewidth]{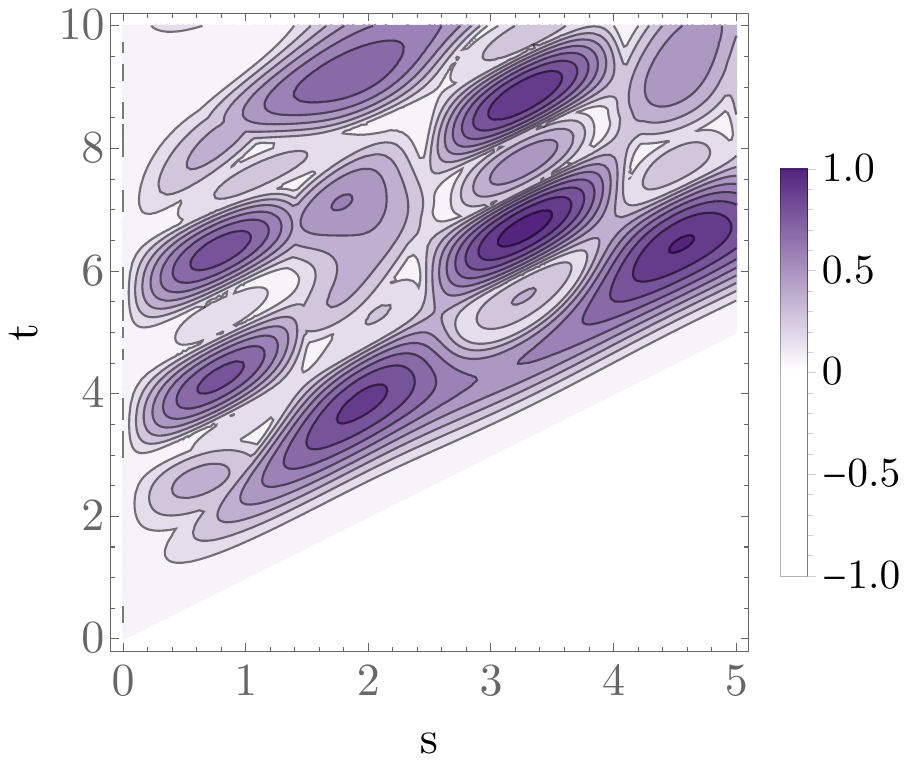}
    \caption{Quantum information backflow (positive values) in the Jaynes-Cummings model. For the initial state, we set the cavity mode in the Fock state $\ket{1}$ and the atom in the maximally mixed state.}
    \label{fig:JC_quantum_info}
\end{figure}

Next, we consider a trace-and-prepare channel as an input for the dynamical supermap,  
and then quantify information backflow based on the signalling power of the transformed channel $M_{t,s}:=T_{t,s}\star (\mathrm{id}\otimes \varpi)$. In some aspects, 
such an information backflow is more informative than the ordinary one defined in 
Eq.~\eqref{eq:no_info_backflow}. First, it detects memory even when there is no 
CP-divisibility breaking, and it also allows the detection of quantum memory. The later follows 
from Observation~\ref{prop:entanglement_breaking}. In fact, in the absence of 
quantum memory in the environment, the output channel must be entanglement-breaking 
and thus
\begin{equation}
\mathrm{S}(M_{t,\hspace{0.025cm}s})\leq \log_2(\dim\mathcal{H}_{\rm OQS}).
\end{equation}
Consequently, any violation of this bound is a signature of quantum memory, which 
we interpret as \emph{quantum information backflow}. In Fig.~\ref{fig:JC_quantum_info} 
we plot this function for the well-known Jaynes-Cummings model, describing the evolution of a two-level atom coupled to a single cavity mode (the atom plays the role of the system $\rm S$ while the mode is the environment), with Hamiltonian $H_{\rm JC}(g)=\hbar g(\sigma_+\otimes b+\sigma_-\otimes b^\dagger)$, $b$ being the standard Bosonic lowering operator of the cavity mode, $[b,\,b^\dagger]=\mathrm{id}$, $g$ is the coupling strength, and we consider the resonant case.

{\bf Quantum channels with memory.---} In general, agents can share communication resources that are not describable as an ordinary quantum channel. Indeed, in a two-agent scenario where a causal order is fixed, say that the signalling direction is from Alice to Bob, then the most general communication resource they share is described by a quantum channel with memory~\cite{KW05}.

Quantum channels with memory are in one-to-one correspondence with quantum superchannels with trivial space $\rm \bar{A}$ (see Fig.~\ref{fig:supermap}). Employing the notation of Fig.~\ref{fig:1}, a quantum channel with memory is described by the following CJ operator
\begin{equation}
\Sigma=\tr_{\rm Aux}[(\sigma^{\TT_{\rm Aux}}\otimes \mathrm{id}_{\Ao})(L\otimes \mathrm{id}_{\Ai})].
\end{equation}
Here we return to the initial labelling with inputs and outputs to avoid confusion [cf. Eq.~\eqref{eq:g_born}]. The ``auxiliary'' system $\rm Aux$ plays the role of the memory of the process.

Our key observation here is that $\tr_{\Ai\Bi}\hspace{0.05cm}\Sigma=\mathrm{id}_{\Ao}$, i.e., $\Sigma$ can also be seen as a channel from $\Ao$ to the combined system $\Ai+\Bi$. Therefore, we can define the signalling power of it accordingly,
\begin{equation}\label{eq:signalling_power_qchannel_memory}
\mathrm{S}(\Sigma)=\log_2\max_{W}\tr(\Sigma W),
\end{equation}
where $W\succcurlyeq 0$ and $\tr_{\Ao}\hspace{0.05cm}W=\mathrm{id}_{\Ai \Bi}$. The significance of this quantity is that it quantifies the causal influence Alice's encoding can have on Bob's local statistics. In fact, notice that, from Observation~\ref{obs:1}(a), $\mathrm{S}(\Sigma)$ is null if and only if $\Sigma=\eta\otimes \mathrm{id}_{\rm \Ao}$, $\eta$ being a bipartite state. This corresponds to a common-cause process  \cite{CS16}. In essence, the agents obtain correlated measurement outcomes, but one cannot causally influence events in other's lab.

Within the context of OQS dynamics, the signalling power for quantum channels with memory extends the notion of information backflow [Eq.~\eqref{eq:no_info_backflow}] for the cases where the initial system-environment state is correlated. Consider an initial non-product state $\varrho_{\rm SEnv}$. One encodes some information onto the OQS, this action being described by the local application of a CPTP map. Suppose one encodes a letter $x$, with a pre-selected CPTP map $\mathcal{M}_x$ locally upon the OQS. At a later time $t$, the reduced OQS state reads
\begin{equation}
\mathcal{S}_t[\mathcal{M}_x]=\tr_{\rm Env}\hsp \mathcal{U}_t(\mathcal{M}_x\otimes \mathcal{I}_{\rm Env})(\varrho_{\rm SEnv})
\end{equation}
The signalling power of $\mathcal{S}_t$ [computed via Eq.~\eqref{eq:signalling_power_qchannel_memory}] quantifies the causal influence of the initial encoding on the OQS state at times $t\geq 0$. Consequently, its non-monotonic behaviour can also be seen as a signature of information backflow.

{\bf {Quantum causal networks and beyond.}---} The above reasoning can be applied to quantify information flow between two groups in multi-agent scenarios. In those cases, the corresponding process matrix represents a family of channels, each of which connects the outputs of one group to the inputs of the other (see \hyperlink{AppendixD}{Appendix D} for a discussion). Each representation defines the communication resources between the two chosen groups. The signalling power of the corresponding channel quantifies the information flow between them. This applies to both process having definitive causal order defined through a directed acyclic graph and processes with indefinite causal order. 

Turning our attention back to the two-agent scenario, we can define the signalling power for the generic process matrix (i.e., without necessarily having a defined causal order) analogously to Eq.~\eqref{eq:signalling_power_qchannel_memory}; i.e., looking at it as a quantum channel $\Ao\Bo\to \Ai\Bi$. By doing so, the quantity we obtain coincides (up to the logarithm) with the ``signalling robustness'' introduced in Ref.~\cite{MBC22} in a resource-theoretic approach.

An interesting consequence of our results in this scenario is the possibility of understanding in an operational and intuitive way that, although quantum theory allows for processes with indefinite causal order, this additional resource does not permit causal loops. This fact was already noticed in Ref.~\cite{OCB12}, where the concept of indefinite causal order was formalised. Employing our function $\rm P$ [Eq.~\eqref{eq:P}] and its operational interpretation [Eq.~\eqref{eq:Pantidist}], we can be easily show that if two agents share a valid process matrix then their average exclusion probabilities satisfy
\begin{align}\label{eq:nocausalloops}
&\max\sum_{x\in \mathcal{X}}\mathrm{Pr}(B\neq X|X=x)\nonumber \\ &\quad+\max\sum_{y\in \mathcal{Y}}\mathrm{Pr}(A\neq Y|Y=y)\leq |\mathcal{X}|+|\mathcal{Y}|-1.
\end{align}
Each term corresponds to the average probability of correctly \emph{excluding} at least one alternative in a superdense-coding-like protocol, from Alice to Bob and from Bob to Alice, respectively (notation is the same as in Observation~\ref{obs:2}, $a\leftrightarrow b$ and $x\leftrightarrow y$). In other words, instead of trying to guess which alternative was input, as in standard superdense coding, the agents have the more modest goal of attempting to exclude at least one alternative. The sum is upper-bounded by one, meaning that Alice and Bob cannot, on average, win such a game simultaneously. In particular, if any of them—say, Bob—can conclusively exclude Alice's inputs, then the marginal channel connecting his output to her input is a trace-and-prepare channel, thus not allowing for any signalling. Therefore, the information he gets cannot be sent back to Alice (which could otherwise be used in her encoding) in each round of the game, thus avoiding a paradoxical causal loop.

Another interesting observation is that if we consider $X$ and $Y$ as binary variables and divide both sides of the inequality~\eqref{eq:nocausalloops} by two, the left-hand side coincides with the inequality used in Ref.~\cite{OCB12} to demonstrate indefinite causal order. The right-hand side equals $3/2$, which is twice as large as the classical bound (i.e., assuming a definite causal order) and strictly larger than the quantum bound of $(2+\sqrt{2})/4$~\cite{OCB12}, exactly mimicking the behaviour of statistical correlations in the simplest Bell test~\cite{Brunner2014}. This suggests that the notion of antidistinguishability may be useful for characterising extreme quantum correlations arising from indefinite causal order under physical principles such as absence of causal loops.

\section*{Discussion}

By defining the signalling power of channels and extending it to high-order maps, we achieved a quantitative description of the information flow (signalling) between classical agents in general quantum processes. The concept of information backflow, identified through the breakdown of DPI, underwent significant expansion in various directions. Notably, we extended it to cases involving initial system-environment correlations. Furthermore, our approach enables the discrimination between classical and quantum backflow of information, while also establishes connections with quantum memory effects~\cite{RBL18,GC21}.

Our work paves the way for numerous potential avenues of exploration. To highlight a few, the local detection of system-environment properties in OQS dynamics, such as information (back)flow and quantum memory, could be enhanced with optimal schemes in experimental setups. On one hand, having full knowledge of the process (e.g., via process tomography) would enable the extraction of all relevant properties, while on the other, our results demonstrate that with the assistance of entanglement, a single measurement strategy per time is sufficient. Determining optimal strategies for extracting such properties in physically relevant setups presents an interesting and highly relevant open problem.

Throughout the paper, we discuss some potential connections between our approach to information flow and foundational concepts in quantum theory. First, we highlight the striking analogy between the gap separating bistochastic quantum superchannels and those that prohibit information flow between pre- and post-processing, and the violation of the information causality principle by post-quantum correlations. Second, in a similar vein, our results show that antidistinguishability plays a meaningful role in the formulation of physical principles for quantum mechanics with indefinite causal order: in the simplest scenario, it fully characterises the absence of causal loops. These observations suggest that recasting the information causality principle within this framework can yield novel insights in both directions, and that antidistinguishability is a powerful concept warranting further attention.

\section*{Acknowledgments}

We thank Roope Uola and Kimmo Luoma for discussions and Dariusz Chruściński and Simon Milz for key remarks. This work was supported by the Deutsche Forschungsgemeinschaft (DFG, German
Research Foundation, project numbers 447948357 and 440958198), the Sino-German 
Center for Research Promotion (Project M-0294), and the German Ministry of 
Education and Research (Project QuKuK, BMBF Grant No. 16KIS1618K).
Z.-P. X. acknowledges support from National Natural Science Foundation of China (Grant No. 12305007), Anhui Provincial Natural Science Foundation (Grant No. 2308085QA29). LSVS acknowledges support from the House of Young Talents of the University of Siegen.

\appendix

\section*{Appendix A. Choi-Jamio{\l}kowski isomorphism and link product}\hypertarget{AppendixA}{ }

A quantum channel is the most general deterministic transformation between density operators that is compatible with the probabilistic structure of quantum theory. Mathematically, it can be represented as a complete-positive (CP) and trace-preserving (TP) linear map $\mathcal{N}:\mathrm{L}(\mathcal{H}_{\rm A})\to \mathrm{L}(\mathcal{H}_{\rm B})$. Alternatively, one can represent a channel as an operator $N\in\mathrm{L}(\mathcal{H}_{\rm A}\otimes \mathcal{H}_{\rm B})$ defined via
\begin{equation}\label{eq:CJ1}
N\equiv \mathrm{CJ}(\mathcal{N}):=\sum_{ij}\ket{i}\bra{j}\otimes \mathcal{N}(\ket{i}\bra{j}).
\end{equation}
The mapping $\mathcal{N}\mapsto N$ defines the Choi-Jamio{\l}kowski (CJ) isomorphism \cite{Choi75,J72}. A simple calculation reveals that its inverse equals
\begin{equation}\label{eq:CJ2}
\mathcal{N}(\cdot)=\tr_{\rm A}[(\cdot^\TT\otimes \mathrm{id}_{\rm B})N].
\end{equation}

\begin{rem}
Both Eqs. (\ref{eq:CJ1}) and (\ref{eq:CJ2}) are computed with respect to a fixed orthonormal basis in the input space. We consider this standard basis-dependent definition of the CJ isomorphism. Henceforth, all CJ operators, transpositions, partial transpositions and so on are computed with respect to fixed orthonormal bases, even though we rarely make explicit reference to it.
\end{rem}

\begin{lem}
A linear map $\mathcal{N}:\mathrm{L}(\mathcal{H}_{\rm A})\to \mathrm{L}(\mathcal{H}_{\rm B})$ is CP if and only if its Choi operator $N$ is positive semidefinite, and TP if and only if $\tr_{\rm B}\hspace{0.05cm}N=\mathrm{id}_{\rm A}$.
\end{lem}

The CJ isomorphism thus provide a representation for quantum channels that are particularly powerful if we are doing semidefinite linear programming (SDP) or if we just want to find out whether a map is CP or not.

From now on, in particular when proving our technical results, we intensively use so-called link product introduced in Ref.~\cite{CDP09}, where one can also find a detailed presentation on the significance and main properties of that product. Here, we just summarise the main properties we will use throughout.

\begin{defn}
Let $\mathbf{A}$ and $\mathbf{B}$ be two finite set of indices labelling Hilbert spaces $\mathcal{H}_i$, $i\in\mathbf{A}\cup \mathbf{B}$. The link product between $U\in \mathrm{L}(\otimes_{i\in \mathbf{A}} \mathcal{H}_i)$ and $V\in \mathrm{L}(\otimes_{i\in \mathbf{B}} \mathcal{H}_i)$ is an operator $U\star V\in \mathrm{L}(\otimes_{i\in \mathbf{A}\cup\mathbf{B}\setminus\mathbf{A}\cap\mathbf{B}}\mathcal{H}_{i})$ defined as
\begin{equation}\label{eq:link_prod_defn}
U\star V=\tr_{\mathbf{A}\cap\mathbf{B}}[(U^{\TT_{\mathbf{A}\cap\mathbf{B}}}\otimes \mathrm{id}_{\mathbf{B}\setminus \mathbf{A}})(V\otimes \mathrm{id}_{\mathbf{A}\setminus \mathbf{B}})],
\end{equation}
where $\TT_{\mathbf{A}\cap\mathbf{B}}$ means partial transposition over Hilbert spaces labelled by $\mathbf{A}\cap\mathbf{B}$.
\end{defn}

The main motivation behind link product is to translate composition of quantum transformations to the CJ picture. In fact, notice that the Choi operator of the composition of two linear maps $\mathcal{N}$ and $\mathcal{M}$ equals
\begin{equation}
\mathrm{CJ}(\mathcal{M}\circ\mathcal{N})=\mathrm{CJ}(\mathcal{M})\star \mathrm{CJ}(\mathcal{N}).
\end{equation}
In particular, the action of a quantum channel onto a quantum state can be written as
\begin{equation}
\mathcal{N}(\varrho)=N\star\varrho.
\end{equation}
In the following lemma states the main properties of the link product \cite{CDP09}.

\begin{lem}
Let $\bf A$ and $\bf B$ and $\bf C$ be finite sets of indices labelling Hilbert spaces $\mathcal{H}_i$, $i\in\mathbf{A}\cup \mathbf{B}\cup \mathbf{C}$ and let $U\in\mathrm{L}(\mathcal{H}_{\bf A})$, $V\in\mathrm{L}(\mathcal{H}_{\bf B})$ and $W\in\mathrm{L}(\mathcal{H}_{\bf C})$. The link product has the following properties:
\begin{itemize}
    \item[(1)] if $U$ and $V$ are Hermitian then $U\star V$ is also Hermitian;
    \item[(2)] if $U$ and $V$ are positive semidefinite then $U\star V$ is also positive semidefinite;
    \item[(3)] $U\star V=U_{\rm SWAP}(V\star U)U_{\rm SWAP}^\dagger$, where $U_{\rm SWAP}$ is the unitary \textsc{swap} between the Hilbert spaces labbeled as $\mathbf{A}\setminus\mathbf{B}$ and $\mathbf{B}\setminus\mathbf{A}$;
    \item[(4)] if $\mathbf{A}\cap \mathbf{B}=\emptyset$ then $U\star V=U\otimes V$;
    \item[(5)] if $\mathbf{A}=\mathbf{B}$ then $U\star V=\tr(U^\TT V)$.
    \item[(6)] if $\mathbf{A}\cap\mathbf{B}\cap \mathbf{C}=\emptyset$ then $U\star V\star W=(U\star V)\star W=U\star(V\star W)$.
\end{itemize}
\end{lem}

\begin{rem}
Hereafter, $N\in\mathrm{Ch}_{\rm AB}^\rightarrow$ means that $N\in\mathrm{L}(\mathcal{H}_{\rm A}\otimes\mathcal{H}_{\rm B})$, $N\pos 0$, and $\tr_{\rm B}\hspace{0.05cm}N=\mathrm{id}_{\rm A}$. Similarly, $N\in\mathrm{Ch}_{\rm AB}^\leftarrow$ if and only if $N\in\mathrm{L}(\mathcal{H}_{\rm A}\otimes\mathcal{H}_{\rm B})$, $N\pos 0$, and $\tr_{\rm A}\hspace{0.05cm}N=\mathrm{id}_{\rm B}$.
\end{rem}

\section*{Appendix B. Properties of the signaling power}\hypertarget{AppendixB}{ }

This appendix contains demonstrations of some of the results established in the main text. For the sake of clarity, some Observations are divided into smaller Propositions.

\begin{propsm}\label{prop:1}
$\mathrm{S}(N)=0$ if and only if $N=\mathrm{id}_{\rm A}\otimes \varpi$, $\varpi\in\mathrm{D}(\mathcal{H}_{\rm B})$.
\end{propsm}
\begin{proof}
If $N=\mathrm{id}_{\rm A}\otimes \varpi$, then $\tr(NW)=\tr\varpi=1$ for all $W\in\mathrm{Ch}_{\rm AB}^\leftarrow$, and hence $\mathrm{S}(N)=0$. To prove the converse, first we define $T=N-\mathrm{id}_{\rm A}\otimes \varpi$,
$\varpi=(\dim \mathcal{H}_{\rm A})^{-1}\tr_{\rm A}\hspace{0.05cm}N\in\mathrm{D}(\mathcal{H}_{\rm B})$. $T$ is a self-adjoint operator with null marginals, i.e., $T=T^\dagger$, $\tr_{\rm A}T=0$ and $\tr_{\rm B}T=0$. Next, we define
\begin{equation*}
W_\lambda=\frac{\mathrm{id}_{\rm AB}}{\dim \mathcal{H}_{\rm A}}+\lambda T.
\end{equation*}
Since $T$ has null marginals, $\tr_{\rm A}W_{\lambda}=\mathrm{id}_{\rm B}$. Moreover, there exists $\varepsilon>0$ such that $\lambda\in[-\varepsilon,\varepsilon]$ implies $W_{\lambda}\succcurlyeq 0$. Therefore, for those values of $\lambda$, $W_{\lambda}\in \mathrm{Ch}_{\rm AB}^\leftarrow$. Hence,
\begin{equation*}
2^{\mathrm{S}(N)}=\max_{W\in\mathrm{Ch}_{\rm AB}^\leftarrow} \tr(NW)\geq \tr(NW_\lambda)=1+\lambda \tr T^2.
\end{equation*}
Finally, $\mathrm{S}(N)=0$ implies $\lambda \tr T^2\leq 0$ for all $\lambda\in[-\varepsilon,\varepsilon]$, consequently $T=0$; that is, $N=\mathrm{id}_{\rm A}\otimes \varpi$.
\end{proof}

\begin{propsm}\label{prop:2}
$\rm S$ is continuous with respect to the diamond norm.
\end{propsm}
\begin{proof}
Since $x,y\geq 1$ implies $|\log x-\log y|\leq |x-y|$, we have
\begin{align*}
|\mathrm{S}(N)-\mathrm{S}(M)|\leq\left|\max_W \tr(NW)-\max_{\hat{W}}\tr(M\hat{W})\right|,
\end{align*}
where the positivity of $\rm S$ was used. Next, we assume, without loss of generality,
\begin{equation*}
\max_W \tr(NW)\geq\max_{\hat{W}}\tr(M\hat{W}),
\end{equation*}
so that  
\begin{align*}
&\left|\max_W \tr(NW)-\max_{\hat{W}}\tr(M\hat{W})\right|&\nonumber \\ &\leq \left|\max_W \tr(NW)-\tr\left(M\argmax_W[\tr(NW)]\right)\right| \nonumber \\ &=\left|\tr\left\{(N-M)\argmax_W[\tr(NW)]\right\}\right|.
\end{align*}
Following the same construction for the proof of Proposition \ref{prop:3} bellow, we can define a set of bipartite states $\{\varrho_x\}_{x\in\mathcal{X}}\subset \mathrm{D}(\mathcal{H}_{\rm A}\otimes \mathcal{H}_{\rm \bar{A}})$, $\mathcal{H}_{\rm \bar{A}}\simeq \mathcal{H}_{\rm A}$, and a POVM $\{E_x\}_{x\in\mathcal{X}}$ such that
\begin{align*}
&\left|\tr\left\{(N-M)\argmax_W[\tr(NW)]\right\}\right|\nonumber \\&\leq \left|\sum_{x\in\mathcal{X}}\tr\big\{E_x[(\mathcal{N}-\mathcal{M})\otimes \mathcal{I}_{\rm \bar{A}}](\varrho_x)\big\}\right|\nonumber \\ &\leq\sum_x\big|\tr\big\{E_x[(\mathcal{N}-\mathcal{M})\otimes \mathcal{I}_{\rm \bar{A}}](\varrho_x)\big\}\big| \nonumber \\ &\leq |\mathcal{X}|\max_{\varrho\in \mathrm{D}(\mathcal{H}_{\rm A}\otimes \mathcal{H}_{\rm \bar{A}}) }\max_{0\leq E\leq \mathrm{id}}\tr\big\{E[(\mathcal{N}-\mathcal{M})\otimes \mathcal{I}_{\rm \bar{A}}](\varrho)\big\} \nonumber \\ &=|\mathcal{X}|\hspace{0.05cm}\|\mathcal{N}-\mathcal{M}\|_{\diamond},
\end{align*}
that is, $|\mathrm{S}(N)-\mathrm{S}(M)|\leq |\mathcal{X}|\hspace{0.05cm}\|\mathcal{N}-\mathcal{M}\|_{\diamond}$.
\end{proof}

\begin{propsm}\label{prop:3}
For channel with CJ operator $N$, used for communication in the superdense coding protocol assisted by a maximally entangled state, one has
\begin{equation}
2^{\mathrm{S}(N)}=\max\sum_{x\in \mathcal{X}}\mathrm{Pr}(B=X|X=x),
\end{equation}
where $|\mathcal{X}|=(\dim\mathcal{H}_{\rm A})^2$, $\mathrm{Pr}(B=X|X=x)$ is the probability of correct decoding the input $x$ and the maximum is taken over all possible encoding and decoding strategies.
\end{propsm}

\begin{proof}
By definition, the sum of the coincidence probabilities in the superdense coding protocol, assisted by a maximally entangled state $\Phi^+\in \mathrm{D}(\mathcal{H}_{\rm \bar{A}}\otimes\mathcal{H}_{\rm \bar{B}})$, $\mathcal{H}_{\rm \bar{A}}\simeq \mathcal{H}_{\rm \bar{B}}$, and with communication channel represented by a CPTP linear map $\mathcal{N}:\mathrm{L}(\mathcal{H}_{\rm A})\to \mathrm{L}(\mathcal{H}_{\rm B})$ is, by definition,
\begin{align}\label{eq:coincidence_prob_proof}
&\sum_{x\in\mathcal{X}}\mathrm{Pr}(B=X|X=x)\nonumber \\ &=\sum_{x\in\mathcal{X}}\tr\left\{E_x\hspace{0.05cm}\left(\mathcal{N}\circ \mathcal{M}_x\right)\otimes \mathcal{I}_{\rm \bar{B}}[\Phi^+]\right\}.
\end{align}
Here $\mathcal{M}_x\in\mathrm{Ch}_{\rm \bar{A}A}^{\rightarrow}$ are the encoding channels, $E_x\in\mathrm{L}(\mathcal{H}_{\rm B}\otimes \mathcal{H}_{\rm \bar{B}})$ are POVM effects of the decoding measurement and $\mathcal{I}_{\rm \bar{B}}$ denotes the identity map acting upon $\mathrm{\bar{B}}$. Invoking the CJ isomorphism and the link product we can write Eq.~\eqref{eq:coincidence_prob_proof} as
\begin{align*}
\sum_{x\in\mathcal{X}}\mathrm{Pr}(B=X|X=x)&=\sum_{x\in\mathcal{X}}E_x\star N\star M_x\star \Phi^+ \nonumber \\ &=N\star \sum_{x\in\mathcal{X}}M_x\star \Phi^+ \star E_x \nonumber \\ &=N\star \tilde{W},
\end{align*}
where $N$ and $M_x$ are CJ representations of $\mathcal{N}$ and $\mathcal{M}_x$, respectively, and we have defined
\begin{equation*}
\tilde{W}=\sum_{x\in \mathcal{X}} M_x \star \Phi^+\star E_x.
\end{equation*}
We now notice that $\tilde{W}\in\mathrm{L}(\mathcal{H}_{\rm A}\otimes \mathcal{H}_{\rm B})$, $\tilde{W}\pos 0$ and
\begin{align*}
\tr_{\rm A}\tilde{W}&=\mathrm{id}_{\rm A}\star \sum_{x\in \mathcal{X}}\hspace{0.05cm} M_x \star \Phi^+\star E_x \nonumber \\ &=\frac{1}{\dim\mathcal{H}_{\rm A}}\sum_{x\in \mathcal{X}} \tr_{\rm \bar{B}}E_x
\nonumber \\ &=\mathrm{id}_{\rm B},
\end{align*}
where the fact that $\{E_x\}$ is a POVM was used, in particular $\sum_{x\in\mathcal{X}} E_x=\mathrm{id}_{\rm B\bar{B}}$. Therefore, $\tilde{W}\in \mathrm{Ch}_{\rm AB}^\leftarrow$, and hence
\begin{equation}\label{eq:max_leq_exp}
\max\sum_{x\in\mathcal{X}}\mathrm{Pr}(B=X|X=x)\leq 2^{\mathrm{S}(N)}.
\end{equation}
To prove the reverse inequality, and thus establish the equality, we first define
\begin{equation*}
W=\argmax_{\hat{W}\in \mathrm{Ch}_{\rm AB}^\leftarrow}\tr N\hat{W}
\end{equation*}
and fix an orthogonal basis of maximally entangled states which, up to normalisation, represent unitary channels; that is, let $u_x\in\mathrm{Ch}_{\rm A\bar{B}}^\rightarrow$, $\tr(u_\mu u_\nu)=(\dim\mathcal{H}_{\rm A})^2\delta_{\mu\nu}$, and
\begin{equation*}
\frac{1}{\dim\mathcal{H}_{\rm A}}\sum_{x\in\mathcal{X}}u_x=\mathrm{id}_{\rm A\bar{B}} .
\end{equation*}
Next, we define $\hat{u}_x\in\mathrm{Ch}_{\rm A\bar{A}}^\rightarrow$, which, for each $x\in\mathcal{X}$, represents the inverse unitary transformation represented by $u_x$. For each $x\in\mathcal{X}$, $\hat{u}_x\star u_x$ represents the identity channel and, hence,
\begin{align}\label{eq:decomp}
\tr(NW)&=N\star W \nonumber \\ &=N\star \hat{u}_x\star u_x \star W \nonumber\\ &=\frac{1}{(\dim\mathcal{H}_{\rm A})^2}\sum_{x\in\mathcal{X}}N\star \hat{u}_x\star u_x \star W.
\end{align}
Now we define $E_x=(\dim\mathcal{H}_{\rm A})^{-1}(u_x\star W)$. For all $x\in\mathcal{X}$, $E_x\pos 0$ and
\begin{align*}
\sum_{x\in\mathcal{X}}E_x&=\frac{1}{\dim\mathcal{H}_{\rm A}}\sum_{x\in\mathcal{X}}u_x\star W \nonumber \\&=\mathrm{id}_{\rm \bar{B}}\otimes \tr_{\rm A}W\nonumber \\&=\mathrm{id}_{\rm B\bar{B}}.
\end{align*}
That is, $\{E_x\}$ is a POVM. The other term in Eq. (\ref{eq:decomp}), $(\dim\mathcal{H}_{\rm A})^{-1}(N\star \hat{u}_x)$, is the state obtained from the application of the local unitary channel represented by $\hat{u}_x$, say $\hat{\mathcal{U}}_x$, to one part of the maximally entangled state $\psi$ followed by application of the channel represented by $N$, 
\begin{equation*}
\frac{N\star \hat{u}_x}{\dim\mathcal{H}_{\rm A}}=\big((\mathcal{N}\circ\hat{\mathcal{U}}_x)\otimes \mathcal{I}_{\rm \bar{B}}\big)[\psi].
\end{equation*}
Therefore, $\{\hat{\mathcal{U}}_x\}$ (encoding) together with $\{E_x\}$ (decoding) define a valid strategy for the superdense coding with coincidence probabilities given by
\begin{equation*}
\mathrm{Pr}(B=X|X=x)=\frac{N\star \hat{u}_x\star u_x\star W}{(\dim\mathcal{H}_{\rm A})^2}.
\end{equation*}
Consequently,
\begin{equation}
2^{\mathrm{S}(N)}\leq \max \sum_{x\in\mathcal{X}}\mathrm{Pr}(B=X|X=x).
\end{equation}
This inequality together with Eq.~\eqref{eq:max_leq_exp} lead to the equality of Eq.~\eqref{eq:coincidence_prob_proof}.
\end{proof}

\begin{prop}\label{prop:4}
For entanglement-breaking channels, the signaling power is upper bounded by $\log_2\dim\mathcal{H}_{\rm A}$.
\end{prop}
\begin{proof}
Every entanglement-breaking channel can be written in its Holevo form \cite{HSR03}, i.e., written as a measure-and-prepare channel,
\begin{equation}
N=\sum_x E_x\otimes\varrho_x,
\end{equation}
where $\{E_x\}$ is a POVM and $\varrho_x$ are density operators. Hence,
\begin{align}
\max_W\tr(NW)&=\max_W\sum_x\tr[(E_x\otimes \varrho_x)^\TT W] \nonumber \\ &=\max_W\sum_x \tr[E_x(W\star \varrho_x)]\nonumber \\ &\leq \sum_x \tr E_x=\dim \mathcal{H}_{\rm A},
\end{align}
and the inequality $\mathrm{S}(N)\leq \log_2\dim \mathcal{H}_{\rm A}$ follows from the monotonicity of the logarithm.
\end{proof}

\begin{propsm}\label{prop:5}
For a given superchannel with CJ operator $T$, the following are equivalent:
\begin{itemize}
    \item[(1)] $\mathrm{S}(T\star N)\leq \mathrm{S}(N)$ for all input channel $N$. 
    \item[(2)] It maps trace-and-prepare channels into trace-and-prepare channels.
    \item[(3)] $T$ is bistochastic.
\end{itemize}
\end{propsm}
\begin{proof}
Clearly, (1) $\implies$ (2). Mathematically, (2) can be stated as $T\star(\mathrm{id}_{\rm A}\otimes \varpi)=\mathrm{id}_{\rm \bar{A}}\otimes \sigma$, $\sigma\in\mathrm{D}(\mathcal{H}_{\rm \bar{B}})$, for all $\varpi\in\mathrm{D}(\mathcal{H}_{\rm B})$ or, equivalently, $\tr_{\rm A}\hspace{0.025cm}(T\star \varpi)=\mathrm{id}_{\rm \bar{A}}\otimes \sigma$. Invoking the result of Ref.~ \cite{CDP09} (Theorem 3 there), one has that, for each $\varpi$, $T\star \varpi$ represents a channel with memory that maps channels $\rm \bar{B}\to \bar{A}$ into states for the system $\rm A$. Since this should hold true for all input states $\varpi$, we conclude that $T$ represents a superchannel transforming channels $\rm \bar{B}\to \bar{A}$ into channels $\rm B\to A$, i.e., $T$ is bistochastic. Finally, if $T$ is bistochastic then $W\in \mathrm{Ch}_{\rm \bar{A}\bar{B}}^\leftarrow$ implies $T\star W\in \mathrm{Ch}_{\rm AB}^\leftarrow$ and, consequently, (3) $\implies$ (1).
\end{proof}

\begin{propsm}\label{prop:6}
$\mathrm{S}(N\otimes M)=\mathrm{S}(N)+\mathrm{S}(M)$.
\end{propsm}
\begin{proof}
Let $N\in\mathrm{Ch}_{\rm AB}^\rightarrow$ and $M\in\mathrm{Ch}_{\rm \tilde{A}\tilde{B}}^\rightarrow$. Given any $W\in\mathrm{Ch}_{\rm AB}^\leftarrow$ and $\tilde{W}\in\mathrm{Ch}_{\rm \tilde{A}\tilde{B}}^\leftarrow$, one has $W\otimes \tilde{W}\in \mathrm{Ch}_{\rm A\tilde{A}B\tilde{B}}^\leftarrow$. Here, $X\in \mathrm{Ch}_{\rm A\tilde{A}B\tilde{B}}^\leftarrow$ if and only if $X\in \mathrm{L}(\mathcal{H}_{\rm A}\otimes\mathcal{H}_{\rm \tilde{A}}\otimes \mathcal{H}_{\rm B}\otimes \mathcal{H}_{\rm \tilde{B}})$, $X\pos 0$ and $\tr_{\rm A\tilde{A}}\hspace{0.05cm}X=\mathrm{id}_{\rm B\tilde{B}}$. Therefore,
\begin{equation}\label{eq:super_add}
\mathrm{S}(N)+\mathrm{S}(M)\leq \mathrm{S}(N\otimes M).
\end{equation}
To prove the reverse inequality, we first set basis of maximally entangled states, $u_x\in \mathrm{Ch}_{\rm AA^\prime}^\rightarrow$ and $v_y\in \mathrm{Ch}_{\rm A\tilde{A}^\prime}^\rightarrow$, that represent unitary channels, $x\in \mathcal{X}$ and $y\in \mathcal{Y}$. Next, we define
\begin{equation*}
E_{xy}=\frac{(u_x\otimes v_y)\star W}{\dim (\mathcal{H}_{\rm A}\otimes \mathcal{H}_{\rm \tilde{A}})}.
\end{equation*}
Clearly $E_{xy}\succcurlyeq 0$ and
\begin{align*}
\sum_{xy}E_{xy}&=\frac{1}{\dim (\mathcal{H}_{\rm A}\otimes \mathcal{H}_{\rm \tilde{A}})}\left(\sum_{x\in\mathcal{X}} u_x\otimes \sum_{y\in\mathcal{Y}}v_y\right)\star W\nonumber \\ &=\mathrm{id}_{\rm A^\prime}\otimes \mathrm{id}_{\rm \tilde{A}^\prime}\tr_{\rm A\Tilde{A}}\hspace{0.05cm}W \nonumber\\ &=\mathrm{id}_{\rm A^\prime \tilde{A}^\prime B\tilde{B}}.
\end{align*}
If $\hat{u}_x$ and $\hat{v}_x$ represent the inverse of the unitaries represented by $u_x$ and $v_y$, respectively, then
\begin{align}\label{eq:add_rev}
&(N\otimes M)\star W\nonumber \\&=\frac{\sum_{xy}[(N\star \hat{u}_x)\otimes (M\star \hat{v}_y)]\star (u_x\otimes v_y)\star W}{[\dim(\mathcal{H}_{\rm A}\otimes \mathcal{H}_{\rm \tilde{A}})]^2}\nonumber \\ &=\sum_{xy}\tr[E_{xy}(\varrho_x\otimes \varsigma_y)],
\end{align}
where $\varrho_x=(\dim \mathcal{H}_{\rm A})^{-1}N\star \hat{u}_x\in \mathrm{D}(\mathcal{H}_{\rm A}\otimes \mathcal{H}_{\rm A^\prime})$ and $\varsigma_y=(\dim \mathcal{H}_{\rm \tilde{A}})^{-1}M\star \hat{v}_y\in \mathrm{D}(\mathcal{H}_{\rm \tilde{A}}\otimes \mathcal{H}_{\rm \tilde{A}^\prime})$. The right-hand side of Eq. (\ref{eq:add_rev}) can be seen as being (proportional to) the average correct probability of discriminating the set of states $\{\varrho_x\otimes \varsigma_y\}$ with no prior information with the POVM $\{E_{xy}\}$. It is always possible to find local POVMs, say $\{F_x\}$ and $\{G_y\}$, that give rise to a decoding strategy that is at least as good as the one given by $\{E_{xy}\}$. In fact, if $\{F_x\}$ and $\{G_y\}$ are optimal POVM for discriminating the set of states $\{\varrho_x\}$ and $\{\varsigma_y\}$ with no prior information, respectively, then $\sum_x F_x \varrho_x-\varrho_{x^\prime}\succcurlyeq 0$ and $\sum_y G_y \varsigma_y-\sigma_{y^\prime}\pos 0$ for all $x^\prime$ and $y^\prime$  \cite{BC09}. As a consequence, $\sum_{xy} (F_x\otimes G_y)(\varrho_x\otimes \varsigma_y)-\varrho_{x^\prime}\otimes \sigma_{y^\prime}\pos 0$ for all $x^\prime$ and $y^\prime$, which implies that $\{F_x\otimes G_y\}$ is an optimal POVM for discriminating $\{\varrho_x\otimes \varsigma_y\}$ with no prior information. Therefore,
\begin{equation*}
\sum_{xy}\tr[E_{xy}(\varrho_x\otimes \varsigma_y)]\leq \sum_{xy}\tr[(F_x\otimes G_y)(\varrho_x\otimes \varsigma_y)].
\end{equation*}
This, together with the result of Proposition~\ref{prop:3}, leads to
\begin{equation}
\mathrm{S}(N)+\mathrm{S}(M)\geq \mathrm{S}(N\otimes M),
\end{equation}
which, together with Eq. (\ref{eq:super_add}), lead to the additivity of $\rm S$.
\end{proof}

\section*{Appendix C. Information dynamics in phase covariant master equations}\hypertarget{AppendixC}{ }

Here, we analyze the information dynamics of a two-level open quantum system under the action of a general phase-covariant noise. In this scenario, a qubit evolves under the action of a quantum dynamical map that commutes with the group of unitaries generated by $\sigma_z$, $\mathcal{U}_z(g)=\ee^{-\ii g \sigma_z}\cdot \ee^{\ii g \sigma_z}$, $g\in\mathbb{R}$. Equivalently \cite{SKHD16}, the evolution of such a system can be described by a local-in-time master equation that has the form 
\begin{align}
\dot{\varrho}_t=-\frac{\ii\omega(t)}{2}[\sigma_z,\varrho_t]&+\frac{\gamma_+(t)}{2}\left(\sigma_+\varrho_t\sigma_--\frac{1}{2}\{\sigma_+\sigma_-,\varrho_t\}\right)\nonumber \\ &+\frac{\gamma_-(t)}{2}\left(\sigma_-\varrho_t\sigma_+-\frac{1}{2}\{\sigma_-\sigma_+,\varrho_t\}\right)\nonumber \\ &+\frac{\gamma_z(t)}{2}\left(\sigma_z\varrho_t\sigma_z-\varrho_t\right).
\end{align}
In other words, the system is under the action of pure dephasing, energy dissipation, and energy gain with time-dependent rates $\gamma_z(t)$, $\gamma_-(t)$, $\gamma_+(t)$, respectively. 

In the computational basis, one can write the solution of a generic phase-covariant master equation in term of the matrix elements,
\begin{subeqnarray}
\bra{1}\rho_t\ket{1}&=&G(t,0)^2\bra{1}\rho_0\ket{1}+H(t), \\
\bra{0}\rho_t\ket{1}&=&G(t,0)\Gamma_z(t)\Omega(t)\bra{0}\rho_0\ket{1},
\end{subeqnarray}
where we have introduced the functions
\begin{align}
G(t,s)&=\exp\left(-\frac{1}{4}\int_s^t[\gamma_+(\tau)+\gamma_-(\tau)]\dd \tau\right), \\ 
H(t)&=\int_0^t G(t,s)^2 \gamma_+(s)\dd s,\\
\Gamma_z(t)&=\exp\left(-\int_0^t \gamma_z(s)\dd s\right),\\
\Omega(t)&=\exp\left(-\ii \int_0^t 2\omega(s)\dd s\right).
\end{align}
Since the evolution commutes with rotations generated by $\sigma_z$, the last term will not contribute to the memory effects, and thus will neglected throughout; that is, we will fix $\Omega(t)=1$ for all $t\geq 0$.

To evaluate the information (back)flow in such a system, we first compute the CJ operator of the quantum dynamical map in the computational basis,
\begin{equation}
N_t={\small\left(\begin{matrix}
    1-H(t)& 0 & 0 & G(t)\Gamma_z(t) \\ 0& H(t) & 0 & 0\\ 0& 0 & 1-G(t)^2-H(t) & 0\\G(t)\Gamma_z(t)& 0 & 0 & G(t)^2+H(t)
\end{matrix}\right)},
\end{equation}
$G(t,0)\equiv G(t)$. If we write down the matrix representation of $W$, $W=[w_{ij}]$, assuming $\tr_{\rm A}\hspace{0.05cm}W=\mathrm{id}_{\rm B}$, we get
\begin{equation}
\tr(N_t W)=1+G(t)^2(w_{11}+w_{44}-1)+2G(t)\Gamma_z(t)\mathrm{Re}(w_{14}).
\end{equation}
In this case, the optimization yields a closed form solution. To obtain it, it is enough to notice that the only elements appearing above are $w_{11}$, $w_{44}$, and $w_{14}$. Moreover, due to the symmetry, we can set $w_{11} = w_{44} \leq 1$ (since $w_{11} + w_{44} \leq \tr W=2$) and $|w_{14}| \leq \sqrt{w_{11}w_{44}}$. Next, define $f_\alpha(x,y)=x+\alpha y$, where $0\leq x\leq 1$ and $|y|\leq x$. The maximum of this function for $\alpha>0$ equals $1+\alpha$. Consequently,
\begin{equation}
2^{\mathrm{S}(N_t)}=1+G(t)^2+2G(t)\Gamma_z(t),
\end{equation}
since $\alpha=\Gamma_z(t)/G(t)>0$. If $\alpha>1$, the minimum of $f_{\alpha}(x,y)$ equals $1-\alpha$. In that case, for $\mathrm{P}(N_t)=1-\min_W \tr(N_t W)$ we obtain
\begin{equation}
\mathrm{P}(N_t)=2G(t)\Gamma_z(t)-G(t)^2,
\end{equation}
where $\alpha=\Gamma_z(t)/G(t)>1$. Taking the derivative of this expression, we conclude that the condition for no information backflow reads
\begin{equation}\label{eq:no_info_backflow_decay_rates}
\gamma_+(t)+\gamma_-(t)+4\gamma_z(t)\pm[\gamma_+(t)+\gamma_-(t)]\frac{G(t)}{\Gamma_z(t)}\geq 0,
\end{equation}
where the condition with the minus sign is valid only if $\Gamma_z(t)/G(t)>1$. If, instead, $0<\alpha\leq 1$, then the minimum of $f_{\alpha}(x,y)$ is zero and, consequently,
\begin{equation}
\mathrm{P}(N_t)=G(t)^2.
\end{equation}
In this case, no information backflow is equivalent to
\begin{equation}
\gamma_+(t)+\gamma_-(t)\geq 0,
\end{equation}
assuming $\Gamma_z(t)/G(t)\leq 1$.

As mentioned in the main text, the eternally non-Markovian evolution \cite{HCLA14} is obtained when we consider, $\gamma_+(t)+\gamma_-(t)=\gamma>0$, and the left-hand-side of Eq. (\ref{eq:no_info_backflow_decay_rates}) null for all times. Consider
\begin{equation}
\gamma+4\gamma_z(t)-\gamma\exp\left(-\frac{1}{4}\int_0^t[\gamma-4\gamma_z(s)]\dd s\right)=0.
\end{equation}
This is an integral equation for $\gamma_z(t)$. Taking the derivative of it, and keeping the information that $\gamma_z(0)=0$ (which follows from the above equality), one obtains
\begin{equation}
16\dot{\gamma}_z+\gamma^2-16 \gamma_z^2=0.
\end{equation}
This is a simple separable ordinary differential equation. Its integration yields
\begin{equation}
\gamma_z(t)=-\frac{\gamma}{4}\tanh\left(\frac{\gamma t}{4}\right),
\end{equation}
corresponding to ``eternally non-Markovian dynamics'' of Ref. \cite{HCLA14}.

\section*{Appendix D. Information flow in causal networks and beyond}\hypertarget{AppendixD}{ }

We now discuss how to extend the notion of signalling power to scenarios involving $n$ parties [or laboratories(labs)] that exchange and measure quantum systems. Each party opens their laboratory exactly once. If a laboratory acts multiple times, each action is treated as a distinct party with a separate label. Inside each laboratory, we assume that quantum mechanics holds locally. Specifically, party $j$ receives an input system described by a finite-dimensional Hilbert space $\mathcal{H}_j^{\tt i}$ and performs a measurement represented by an arbitrary quantum instrument. In the CJ picture, this instrument is given by a collection of operators $\mathbf{E}_j = \{E_{A_j=a_j}\}$ corresponding to different measurement outcomes. After the measurement, party $j$ outputs a (potentially different) quantum system in a post-measurement state associated with the Hilbert space $\mathcal{H}_j^{\tt o}$.

If we assume local linearity of the measurement operators (in order to ensure convex linearity of probabilities under local mixtures of measurements), the joint probability distribution should be determined as in
\begin{align}
&\mathrm{Pr}(A_1\dots A_n\mid \mathbf{E}_1\dots \mathbf{E}_n)\nonumber \\&\quad = \tr\big[\big(E_{A_1}^\mathrm{T} \otimes \dots \otimes E_{A_n}^\mathrm{T}\big)(\Upsilon_n \otimes \varrho)\big].
\end{align}
Here, $\Upsilon_n \in \mathrm{L}(\mathcal{H}_{1}^{\tt i} \otimes \mathcal{H}_{1}^{\tt o} \otimes \dots \otimes \mathcal{H}_{n}^{\tt i} \otimes \mathcal{H}_{n}^{\tt o})$ is the so-called process matrix, and $\varrho \in \mathrm{D}(\mathcal{H}_{1}^{\rm ns} \otimes \dots \otimes \mathcal{H}_{n}^{\rm ns})$ is a multipartite quantum state. For a process with definite causal order, one can always define an enumeration $[n]$ such that $i \leq j$ implies that the $j$-th lab cannot signal to the $i$-th lab. That is, the marginal distribution
\begin{align*}
\sum_{a_j} \mathrm{Pr}(A_1\dots A_i\dots A_j=a_j\dots A_n|\mathbf{E}_1\dots \mathbf{E}_i\dots \mathbf{E}_j\dots \mathbf{E}_n)
\end{align*}
is independent of the instrument $\mathbf{E}_j$. These constraints are encoded in $\Upsilon_n$ as a hierarchy of trace conditions~\cite{CDP09}:
\begin{equation}
\tr_{j}^{\tt i} \hspace{0.05cm} \Upsilon_j = \mathrm{id}_{j-1}^{\tt o} \otimes \Upsilon_{j-1} \hspace{0.5cm} \text{for } 2 \leq j \leq n,
\end{equation}
where $\Upsilon_j$ are recursively defined, starting with $\Upsilon_1 \in \mathrm{D}(\mathcal{H}_{1}^{\tt i})$. Since the $n$-th lab cannot signal to any other, we may, without loss of generality, take $\mathcal{H}_{n}^{\tt o} \simeq \mathbb{C}$.

In the bipartite case, with labs being named Alice (A) and Bob (B), we have $\Upsilon_2\in \mathrm{L}(\mathcal{H}_{\rm A}^{\tt i}\otimes\mathcal{H}_{\rm A}^{\tt o}\otimes \mathcal{H}_{\rm B}^{\tt i})$, $\Upsilon_2\pos 0$ and
\begin{equation}
\tr_{\rm B}^{\tt i}\hspace{0.05cm}\Upsilon_2=\varpi\otimes \mathrm{id}_{\rm A}^{\tt o},
\end{equation}
where $\varpi\in \mathrm{D}(\mathcal{H}_{\rm A}^{\tt i})$. This operator represents an admissible transformation that maps quantum channels $\mathcal{M}:\mathrm{L}(\mathcal{H}_{\rm A}^{\tt i})\to \mathrm{L}(\mathcal{H}_{\rm A}^{\tt o})$ into quantum states $\eta \in\mathrm{D}(\mathcal{H}_{\rm B}^{\tt i})$ as well as a quantum channel with input being the outputs of Alice. In fact, notice that $\tr_{\rm A}^{\tt i}\circ\tr_{\rm B}^{\tt i} \hspace{0.05cm}\Upsilon_2=\mathrm{id}_{\rm A}^{\tt o}$. Based on this observation one can define
\begin{equation}
\mathrm{S}(\Upsilon_2)=\log\max_W \tr(\Upsilon_2W),
\end{equation}
where the maximum is taken over positive semidefinite operators satisfying $\tr_{\rm A}^{\rm o} W=\mathrm{id}_{\rm AB}^{\tt i}$. In the main text, we observe that
\begin{equation}
    \mathrm{S}(\Upsilon_2)=0 \hspace{0.5cm} \iff \hspace{0.5cm} \Upsilon_2=\eta\otimes\mathrm{id}_{\rm A}^{\tt o}, 
\end{equation}
where $\eta\in\mathrm{D}(\mathcal{H}_{\rm A}^{\tt i}\otimes \mathcal{H}_{\rm B}^{\tt i})$. These process matrices correspond to common cause processes: even though Alice and Bob can get correlated outcomes, there is no causal influence from $\rm A$ to $\rm B$. In fact, notice that within this process, Alice cannot signal Bob, i.e., effectively they only share non-signalling resources. 

\begin{figure}
    \centering
    \includegraphics[scale=0.65]{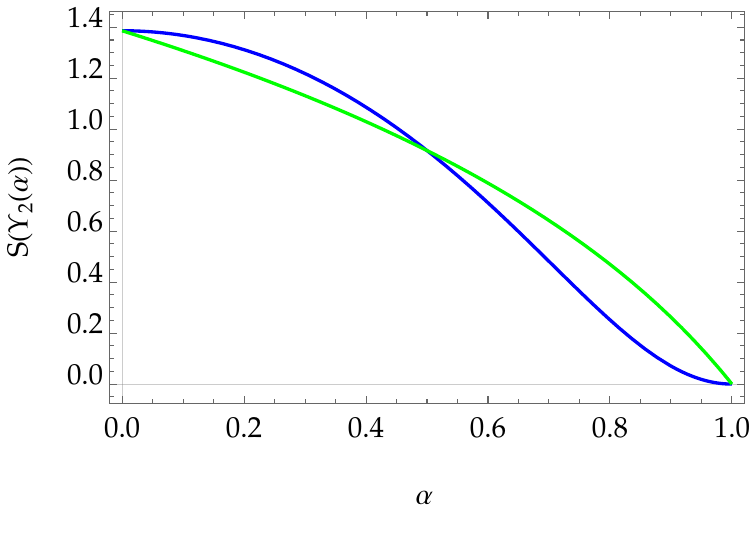}
    \caption{Direct causal influence in processes consisting of coherent (blue) and incoherent (green) mixtures of common-cause and direct-cause.}
    \label{fig:coherent_vs_incoherent}
\end{figure}

As an example, in Fig.~\ref{fig:coherent_vs_incoherent} we plot this function for processes consisting of coherent and incoherent superposition of common cause, i.e., $\Upsilon_2^{\rm cc}=\eta\otimes\mathrm{id}_{\rm A}^{\tt o}$, and direct cause, $\Upsilon_2^{\rm dc}=\varpi\otimes K$, where $K\in\mathrm{L}(\mathcal{H}_{\rm A}^{\tt o}\otimes\mathcal{H}_{\rm B}^{\tt i})$, $K\succcurlyeq 0$ and $\tr_{\rm B}^{\tt i} \hspace{0.05cm}K=\mathrm{id}_{\rm A}^{\tt o}$. The processes we use to plot Fig.~\ref{fig:coherent_vs_incoherent} was the one considered in Ref.~\cite{MRSR17}. The common-cause and direct-cause processes are, respectively, $\Upsilon_2^{\rm cc}=\ket{\Phi_2^+}\bra{\Phi_2^+}\otimes \mathrm{id}_{\rm A}^{\tt o}$ and $\Upsilon_2^{\rm dc}=\mathrm{id}_{\rm A}^{\tt i}\otimes \ket{\Phi_2^+}\bra{\Phi_2^+}$, where $\sqrt{2}\ket{\Phi_2^+}=\ket{00}+\ket{11}$. An incoherent mixture between common-cause and direct-cause is a process of the form
\begin{equation}
\Upsilon_2^{\rm incoh}(\alpha)=\alpha\Upsilon_2^{\rm cc}+(1-\alpha)\Upsilon_2^{\rm dc}.
\end{equation}
To define a coherent mixture of causal relations, one can employ the partial \textsc{swap} unitary,
\begin{equation}
U_{\rm PS}(\theta)=\cos\left(\frac{\theta}{2}\right)\mathrm{id}+\ii \sin\left(\frac{\theta}{2}\right)U_{\rm SWAP},
\end{equation}
We then define
\begin{equation}
\Upsilon_2^{\rm coh}(\alpha)=\Phi_2^+ \star \|U_{\rm {PS}}(\pi\alpha)\rangle\langle U_{\rm {PS}}(\pi\alpha)\|\star \mathrm{id}_{\rm Aux^\prime},
\end{equation}
where $\|U_{\rm {PS}}(\pi\alpha)\rangle$ is the vectorization of the unitary $U_{\rm {PS}}(\pi\alpha)$. Clearly, $\Upsilon_2^{\rm coh}(0)=\Upsilon_2^{\rm dc}$ and $\Upsilon_2^{\rm coh}(1)=\Upsilon_2^{\rm cc}$.

For processes with indefinite causal order, one can also state a similar observation. In particular, consider a bipartite process matrix $\Upsilon_2$, where now we do not assume an external causal structure. The constrains on this operator can be written as (see Ref.~\cite{ABCFGB15} for details)
\begin{subeqnarray}\label{eq:pm_constr}
\Upsilon_2={}_{\Ao} \Upsilon_2+{}_{\Bo} \Upsilon_2-{}_{\Ao\Bo}\Upsilon_2,\\ 
{}_{\Ai\Ao} \Upsilon_2={}_{\Ai\Ao\Bo} \Upsilon_2, \\
{}_{\Bi\Bo} \Upsilon_2={}_{\Bi\Bo\Ao} \Upsilon_2,
\end{subeqnarray}
together with $\Upsilon_2\pos 0$ and $\tr \Upsilon_2=\dim(\mathcal{H}_{\Ao}\otimes \mathcal{H}_{\Bo})$. Here, we use the notation introduced in Ref.~\cite{ABCFGB15}, namely ${}_X W=(\dim\mathcal{H}_X)^{-1}\tr_{X} W\otimes \mathrm{id}_X$. 

Now notice that $\tr_{\Ai\Bi} \Upsilon_2=\mathrm{id}_{\Ao\Bo}$, therefore, $\Upsilon_2$ also represents a channel connecting outputs to inputs. A consequence of this observation is that one can state the fact already notice in Ref.~\cite{OCB12} that no causal loops are allowed in a process satisfying Eq.~(\ref{eq:pm_constr}) with an ``operational flavor'' as discussed in the main text. This is possible because Eq.~\eqref{eq:pm_constr} implies
\begin{equation}\label{eq:nocl}
\mathrm{P}\left(\frac{\tr_{\Bi\Ao} \Upsilon_2}{\dim\mathcal{H}_{\Bo}}\right)+\mathrm{P}\left(\frac{\tr_{\Ai\Bo} \Upsilon_2}{\dim\mathcal{H}_{\Ao}}\right)\leq 1,
\end{equation}
which leads to our interpretation in the main text. This inequality is derived as follows. First, we multiply both sides of the first condition in Eq.~\eqref{eq:pm_constr} by $W\otimes \tilde{W}$, where $W\in \mathrm{Ch}_{\Bi \Ao}^\rightarrow$ and $\tilde{W}\in \mathrm{Ch}_{\Ai \Bo}^\rightarrow$ arbitrary. Next, we use the other two conditions together with the definition of $\rm P$. Finally, the inequality follows from the positivity of the process matrix $\Upsilon_2$. The inequality \eqref{eq:nocausalloops} in the main text is equivalent to Eq.~\eqref{eq:nocl}, when we use our construction of Proposition~\ref{prop:3} replacing maximum by minimum for the corresponding marginal channels.

\bibliography{main}

\end{document}